\documentclass[pra,superscriptaddress,aps,10pt,onecolumn,nofootinbib]{revtex4}
\usepackage{amsmath,amssymb,graphicx,enumerate,bbm,color}
\usepackage{amsthm}
\usepackage{hyperref}

\newtheorem{thm}{Theorem}

\newtheorem*{cor}{Corollary}
\providecommand{\bra}[1]{{\langle#1|}}
\providecommand{\ket}[1]{{|#1\rangle}}
\providecommand{\braket}[2]{{\langle#1|#2\rangle}}
\providecommand{\ketbra}[2]{{|#1\rangle\!\langle#2|}}
\providecommand{\sandwich}[3]{{\langle#1|#2|#3\rangle}}

\newcommand{\one}{\mathbbm{1}}
\newcommand{\tr}{{\rm tr}}

\newcommand{\Phiplus}{\Phi^{\!+}}
\definecolor{nred}{rgb}{0.9,0.1,0.1}
\definecolor{nblack}{rgb}{0,0,0}
\definecolor{nblue}{rgb}{0.2,0.2,0.8}
\definecolor{ngreen}{rgb}{0.2,0.6,0.2}

\newcommand{\POVMA}{ \mathcal{A}^\text{A}_{a|s} }
\newcommand{\POVMB}{ \mathcal{B}^\text{B}_{b|t}}
\newcommand{\SqPOVMA}{ \mathcal{A}^\textup{A\!'A}_{a} }
\newcommand{\SqPOVMB}{ \mathcal{B}^\textup{BB'}_{b}}

\begin{document}

\title{Entangled states cannot be classically simulated\\ in generalized Bell experiments with quantum inputs}

\author{Denis Rosset}
\affiliation{Group of Applied Physics, University of Geneva, Switzerland}
\author{Cyril Branciard}
\affiliation{Centre for Engineered Quantum Systems and School of Mathematics and Physics, The University of Queensland, St Lucia, QLD 4072, Australia}
\author{Nicolas Gisin}
\affiliation{Group of Applied Physics, University of Geneva, Switzerland}
\author{Yeong-Cherng Liang}
\affiliation{Group of Applied Physics, University of Geneva, Switzerland}

\date{\today}

\begin{abstract}
Simulation tasks are insightful tools to compare information-theoretic resources. Considering a generalization of usual Bell scenarios where external quantum inputs are provided to the parties, we show that any entangled quantum state exhibits correlations that cannot be simulated using only shared randomness and classical communication, even when the amount and rounds of classical communication involved are unrestricted. 
We indeed construct explicit Bell-like inequalities that are necessarily satisfied by such classical resources but nevertheless violated by correlations obtainable from entangled quantum states, when measured a single copy at a time.
\end{abstract}

\maketitle

\section{Introduction}

Understanding how quantum resources compare to classical ones is of fundamental importance in the rising field of quantum information science~\cite{NielsenChuang}.
An insightful approach is to characterize, for instance, information processing tasks that can be achieved by distributed parties sharing different resources~\cite{qms,Buhrman:2010nonlocality}. Notable examples of classical resources are shared randomness and classical communication of various types (distinguished, e.g., by restrictions on the amount and/or direction of communication, and/or number of rounds performed~\cite{Buhrman:2010nonlocality,Chitambar:2012te}). Quantum resources also fall under different categories, such as separable or entangled states~\cite{Horodecki:2009gb} in particular.

A natural way to compare quantum resources against classical ones is via the simulation~\cite{Maudlin:1992bell} of {\it  quantum measurement scenarios}~\cite{qms}. There, one tries to simulate the correlations obtained by measuring a given quantum state in a Bell scenario~\cite{Bell:1964,Bell:1987speakable}, allowing the parties to use in the simulation shared randomness and possibly classical communication~\cite{Brassard:2003bc}. If no communication is allowed, the set of simulable correlations is bounded by Bell inequalities~\cite{Pitowsky:1989quantum}, whereas an infinite amount of communication allows the simulation of any correlations. Conversely, the number of classical bits exchanged by the parties can be used to quantify the entanglement of a given state that violates some Bell inequalities. For the example of the singlet state, projective measurements can be simulated with 1 bit in the {\it  worst case} scenario~\cite{Toner:2003communication}, and more general measurements with 6 bits {\it  on average}~\cite{Degorre:2005simulating}. A more general model by Massar {\it  et al.}~\cite{Massar:2001ug} simulates arbitrary measurements on arbitrary bipartite states using a finite number of bits on average~\cite{Regev:2010hp,Vertesi:2009ew}. A few results are known for the multipartite case: for instance, projective equatorial measurements on the multipartite Greenberger-Horne-Zeilinger state~\cite{Greenberger:1989going} can be simulated using shared randomness and finite communication on average~\cite{Bancal:2010ho}; in the tripartite case, 3 bits are always sufficient~\cite{Branciard:2011dx}.

Surprisingly, only small amounts of communication are necessary in the tasks mentioned above, which we will refer to as {\it  standard simulation tasks}. In fact, even the correlations obtained from entanglement swapping~\cite{Zukowski:1993fs}(a more complex process involving both entangled states and entangled measurements) can be simulated using only finite communication and uncorrelated shared randomness~\cite{Branciard:1429040}. Of course, it should not be forgotten that there exist entangled states which --- when measured one copy at a time --- can be simulated using solely shared randomness (see, e.g. Refs.~\cite{Werner:1989zz,Barrett:2002gu}); such states cannot be distinguished from separable states in a standard simulation task. 
Nevertheless, these states have an advantage over separable states for quantum teleportation{~\cite{Popescu:1994bell}}, and their nonlocal behavior can be demonstrated using more subtle Bell tests (see, e.g., Refs.~\cite{Masanes:2008hidden,Liang:2012} and references therein). Note that all these tasks involve single copies of quantum states. Other scenarios allowing joint measurement on multiple copies of a quantum state~\cite{Liang:2006fe,Palazuelos:2012fz,Cavalcanti:2012vg} nevertheless enable states with Bell-local correlations to exhibit non-local correlations. We will however not consider such possibilities in our paper, and stick to the problem of simulating correlations produced by single copies of quantum states.

\begin{figure}
  \includegraphics{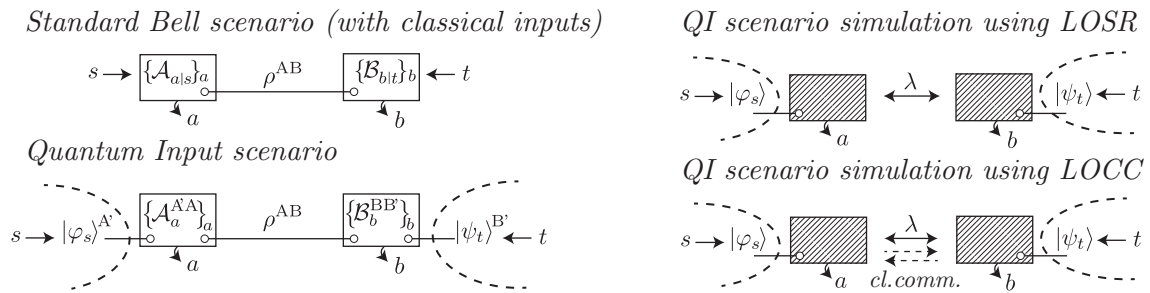}
  \caption{\label{Fig:scenarios}In a standard bipartite Bell scenario, Alice and Bob receive or select randomly classical inputs $s$, $t$ indicating measurements to be performed on the shared state $\rho^\text{AB}$, with outcomes $a$, $b$.
    In a quantum input (QI) scenario, the inputs are quantum states $\ket{\varphi_s}$, $\ket{\psi_t}$, prepared and provided by an external agent, to be measured jointly but locally by Alice and Bob with their respective part of the state $\rho^\text{AB}$, producing outcomes $a$, $b$. It is assumed that there is no leak of information about the indices $s$, $t$.
    We then consider the task of simulating QI scenarios using classical resources, namely local operations assisted by shared randomness (LOSR), as well as classical communication (LOCC).}
\end{figure}

Recently, Buscemi~\cite{Buscemi:2012all} generalized standard Bell scenarios by using {\it quantum} inputs, and showed that, in those scenarios, single copies of any entangled state can produce correlations which are non-simulable using only shared randomness and local operations (see Figure~\ref{Fig:scenarios}). Naturally, these scenarios call also for a generalization of the standard simulation task which we will refer to as the {\it quantum input simulation task}. These simulation tasks are defined in Section~\ref{Sec:SimulationTasks}. We prove our main result in Section~\ref{Sec:NonSimulability}, namely that in the quantum input framework, it is {\it impossible} to simulate {\it any} entangled state using shared randomness and even unrestricted classical communication --- thus establishing a new feature of quantum entanglement. We demonstrate this by constructing, for any entangled state, a generalized Bell scenario and an inequality that is necessarily satisfied by classical resources, but which can be violated by the given entangled state. Explicit examples of such inequalities for Werner states are given in Section~\ref{Sec:WernerStates}.

\section{Simulation tasks}
Before proving our main result, let us generalize the standard simulation tasks to allow for quantum inputs. We shall start by recalling from Ref.~\cite{qms} the simulation task inspired from a standard Bell test.
\label{Sec:SimulationTasks}
\subsection{Standard simulation task}
In the bipartite case, we define via the triple $(\rho^\text{AB}, \{\mathcal{A}^\text{A}_{a|s}\},\{\mathcal{B}^\text{B}_{b|t}\})$ the scenario to be simulated, where $\rho^\text{AB}$ is a bipartite entangled state shared between Alice and Bob, while $\{ \mathcal{A}^\text{A}_{a|s}\}$ and $\{ \mathcal{B}^\text{B}_{b|t}\}$ are positive operator-valued measure (POVM) elements~\cite{NielsenChuang} describing the respective measurements to be performed on their subsystems; each POVM, labeled by settings $s$ or $t$, has outcomes $a$ or $b$. The simulation task is defined as follows: Alice and Bob are given beforehand classical descriptions of the triple specifying the scenario and have access to shared randomness. In each round of the simulation, they each receive a classical description of their input $s$ or $t$ but do not know the other party's input. Their task is to produce --- possibly after some rounds of classical communication with each other --- classical outputs $(a,b)$ such that they reproduce exactly the joint conditional probability distribution

\begin{equation}
\label{Eq:P_standard}
P_\rho(a,b|s,t) = \tr \left [ \left ( \mathcal{A}^\text{A}_{a|s} \otimes \mathcal{B}^\text{B}_{b|t} \right ) \cdot \rho^\text{AB} \right ],
\end{equation}
as predicted by quantum mechanics.

It is well known that some quantum correlations can violate Bell inequalities, and as such cannot be simulated using only shared randomness. The canonical no-go result is represented by the Clauser-Horne-Shimony-Holt (CHSH)~\cite{Clauser:1969ff} Bell scenario, where each of $(s,t,a,b)$ is binary. In this scenario, the CHSH inequality defines constraints that have to be satisfied by Bell-local correlations~\cite{Bell:1964,Bell:1987speakable,Norsen:2006di}; these constraints can be maximally violated with judicious choices of measurements $\{\POVMA\}$, $\{\POVMB\}$ when Alice and Bob share for instance a two-qubit singlet state $\rho^\text{AB} = \ketbra{\Psi^-}{\Psi^-}$, with $\ket{\Psi^-} = \frac{1}{\sqrt{2}}(\ket{01}-\ket{10})$.
Note that only entangled states $\rho^\text{AB}$ can generate Bell-nonlocal correlations; entanglement is however not a sufficient resource, as some entangled states can only generate Bell-local correlations when measured one copy at a time~\cite{Werner:1989zz,Barrett:2002gu}.

When any amount of communication is allowed, a simple strategy to achieve the simulation task is the following. Alice communicates the input that she receives to Bob, and then they produce outputs in accordance with some pre-established strategy via shared randomness. However, there may be more efficient strategies when more inputs are considered: for instance, the simulation of all projective measurements on the singlet state can be achieved using only one bit of communication in addition to shared randomness~\cite{Toner:2003communication}.

\subsection{Quantum Input (QI) simulation task} 
Let us now consider a variant inspired by Buscemi's work~\cite{Buscemi:2012all}. The scenario of a quantum-input simulation task is defined via the 5-tuple $(\rho^\text{AB},\{\ket{\varphi_s}^{\text{A'}}\},\{\ket{\psi_t}^{\text{B'}}\},\{\SqPOVMA\},\{\SqPOVMB\})$ where the superscripts identify the relevant subsystems (see Figure~\ref{Fig:scenarios}); the sets $\{\ket{\varphi_s}^{\text{A'}}\}$,$\{\ket{\psi_t}^{\text{B'}}\}$ specify the quantum input states prodived to Alice and Bob, while $\{\SqPOVMA\}$ and $\{\SqPOVMB\}$ are POVM elements acting respectively on the systems A'A and BB'. As with the standard simulation task, Alice and Bob are given beforehand classical descriptions of the 5-tuple specifying the scenario and are assumed to have access to shared randomness.  However, instead of a classical description of the inputs $s$ and $t$, in each round of the simulation task, Alice and Bob now receive respectively from external state preparation devices the quantum states $\ket{\varphi_s}^{\text{A'}}$ and $\ket{\psi_t}^{\text{B'}}$ to be measured jointly with $\rho^\text{AB}$. The goal here is again for Alice and Bob to produce, possibly with the help of classical communication, outputs $a$ and $b$ such that they reproduce exactly the conditional probability distribution

\begin{equation}
\label{Eq:P_QI_complete}
P_\rho(a,b \, | \, \ket{\varphi_s},\ket{\psi_t}) = \, \tr \big [ \! \left ( \mathcal{A}^\text{A\!'A}_a \! \otimes \mathcal{B}^\text{BB'}_b \right ) \cdot \left ( \ketbra{\varphi_s}{\varphi_s}^\text{A\!'} \! \otimes \rho^\text{AB} \! \otimes \ketbra{\psi_t}{\psi_t}^\text{B'} \right ) \! \big ],
\end{equation}
as predicted by quantum mechanics.

The correlations~\eqref{Eq:P_QI_complete} can also be written in a similar form as~\eqref{Eq:P_standard}, namely
\begin{equation}
\label{Eq:P_QI_effective}
P_\rho(a,b \, | \, \ket{\varphi_s},\ket{\psi_t}) = \tr \left [ \left ( \mathcal{A}^\textup{A}_{a|\ket{\varphi_s}} \otimes \mathcal{B}^\textup{B}_{b|\ket{\psi_t}} \right ) \cdot \rho^\text{AB} \right ],
\end{equation}
where the operators
\begin{equation}
\label{Eq:POVM_AB}
	\mathcal{A}^\textup{A}_{a|\ket{\varphi_s}} =  ^\textup{A\!'}\!\bra{\varphi_s} \!\otimes\! \mathbbm{1} \cdot  \mathcal{A}^\textup{A\!'A}_a  \cdot \ket{\varphi_s}^{\!\textup{A\!'}} \otimes\! \mathbbm{1}, \qquad
	\mathcal{B}^\textup{B}_{b|\ket{\psi_t}} = \mathbbm{1} \!\otimes \, ^\textup{B'}\!\!\bra{\psi_t} \cdot  \mathcal{B}^\textup{BB'}_b \cdot \mathbbm{1} \!\otimes\! \ket{\psi_t}^{\!\textup{B'}} 
\end{equation}
describe Alice and Bob's effective POVMs acting on $\rho^\text{AB}$, for each input state $\ket{\varphi_s},\ket{\psi_t}$. Here, in contrast to standard simulation tasks~\eqref{Eq:P_standard}, Alice and Bob do not know {\it a priori} which effective POVMs $\{\mathcal{A}^\textup{A}_{a|\ket{\varphi_s}}\}_a$ and $\{\mathcal{B}^\textup{B}_{b|\ket{\psi_t}}\}_b$ they should simulate, as they do not know (and may not able to determine with certainty) the classical indices $s$,$t$ of their quantum inputs $\ket{\varphi_s}$,$\ket{\psi_t}$: these states are indeed chosen randomly and prepared by external devices they do not control\footnote{This contrasts with usual Bell scenarios, where there is no difference whether indices $s$ and $t$ are provided by an external agent or randomly chosen by the parties--- provided that such a choice is independent of any underlying variables describing the system\cite{Hall:2010hc}.}.

Note however that when all Alice's (respectively Bob's) input states $\{\ket{\varphi_s}\}$ ($\{\ket{\psi_t}\}$) are orthogonal to one another, Alice (Bob) can perfectly distinguish them; then the QI simulation task simply reduces to the standard one. As such, there is a distinction between the two tasks only when the quantum input states are non-orthogonal, and therefore non-distinguishable. Note that the non-orthogonality of the set of quantum input states only implies, according to quantum theory, that the received state cannot be determined with certainty. If the states in $\{\ket{\varphi_s} \}_s$ are non-orthogonal but still linearly independent, they can be unambiguously discriminated with a non-zero probability of getting a conclusive answer~\cite{Chefles:1998hd}.

\section{Non-simulability of entangled quantum states}
\label{Sec:NonSimulability}
We already know from Buscemi's result~\cite{Buscemi:2012all} that all entangled quantum states can produce correlations that are non-simulable when both parties perform only arbitrary local operations assisted by shared randomness (LOSR). Clearly, a natural follow-up question is whether classical communication could help in the simulation of these correlations, in the paradigm of local operations and classical communication (LOCC). We answer below this question by the negative, showing that any entangled state can produce correlations that cannot be reproduced using LOCC. To do so, we first define a canonical QI simulation task relevant for entangled states of dimension $d_A \times d_B$. We then construct, for any such state, an explicit Bell-like inequality violated by the correlations of the given entangled state but otherwise satisfied by LOCC resources.

\subsection{Canonical QI simulation task} 
We consider a state $\rho^\text{AB}$ of dimension $d_A\times d_B$, and prescribe the following: Alice and Bob perform local measurements $\{\SqPOVMA\},\{\SqPOVMB\}$ with binary outcomes $a,b$, where the outcome $a$ or $b=1$ corresponds to the successful projection onto the maximally entangled state $\ket{\Phiplus_d} = \sum_{k=0}^{d-1} \ket{kk} / \sqrt{d}$. Thus:
\begin{eqnarray}
\mathcal{A}^\textup{A\!'A}_1 = \ketbra{\Phiplus_{d_A}}{\Phiplus_{d_A}},  \mathcal{A}^\textup{A\!'A}_0  =  \mathbbm{1} - \mathcal{A}^\textup{A\!'A}_1 , \qquad
\mathcal{B}^\textup{BB'}_1  =  \ketbra{\Phiplus_{d_B}}{\Phiplus_{d_B}}, \mathcal{B}^\textup{BB'}_0  =  \mathbbm{1} - \mathcal{B}^\textup{BB'}_1. \label{Eq:BellMeasurementAB}
\end{eqnarray}

The input states are chosen to have the same dimensions as the respective subsystem of A and B in $\rho^\text{AB}$, $\ket{\varphi_s} \in \mathbbm{C}^{d_A}$, $\ket{\psi_t} \in \mathbbm{C}^{d_B}$, and constructed such that the corresponding density matrices $\{ \ketbra{\varphi_s}{\varphi_s} \}$, $\{ \ketbra{\psi_t}{\psi_t} \}$ span the space of linear operators acting on $\mathbbm{C}^{d_A}$, $\mathbbm{C}^{d_B}$, as is also done in the proof and example of~\cite{Buscemi:2012all}. Such sets can always be constructed using the minimal number $d_A^2, d_B^2$ of elements, with $s=1,\ldots, d_A^2$ and $t=1,\ldots, d_B^2$. Together with $\rho^\text{AB}$, these elements define completely our quantum-input simulation task.

The effective POVMs applied to $\rho^\text{AB}$ are then described [from Eq.~\eqref{Eq:POVM_AB}] by
\begin{eqnarray}
\mathcal{A}^\textup{A}_{1|\ket{\varphi_s}}  \! = \frac{ \ketbra{\varphi_s}{\varphi_s}^{\!\top}}{d_A}, & &
\mathcal{A}^\textup{A}_{0|\ket{\varphi_s}}  \! = \mathbbm{1} - \frac{ \ketbra{\varphi_s}{\varphi_s}^{\!\top}}{d_A}, \label{Eq:POVM_A_canonical}\\
\mathcal{B}^\textup{B}_{1|\ket{\psi_t}}  \! = \frac{ \ketbra{\psi_t}{\psi_t}^{\!\top}}{d_B}, & &
\mathcal{B}^\textup{B}_{0|\ket{\psi_t}}  \! = \mathbbm{1} - \frac{ \ketbra{\psi_t}{\psi_t}^{\!\top}}{d_B}, \label{Eq:POVM_B_canonical}
\end{eqnarray}
where $^{\top}$ indicates transposition in the computational bases $\{ \ket{k} \}$ of $\mathbbm{C}^{d_A}$ and $\mathbbm{C}^{d_B}$.

\subsubsection*{Canonical QI simulation task of the singlet state}
As a concrete example, consider the case where Alice and Bob share the singlet state $\rho^\text{AB} = \ketbra{\Psi^-}{\Psi^-}$. The input states can for instance be chosen as the vertices of a regular tetrahedron on the surface of the Bloch sphere. Using $\vec{\sigma} = (\sigma_1, \sigma_2, \sigma_3)$ the vector of Pauli matrices, and 
$\vec{v}_1 = (1,1,1)/\sqrt{3}$, $\vec{v}_2 = (1,-1,-1) /\sqrt{3}$, $\vec{v}_3 = (-1,1,-1) /\sqrt{3}$, $\vec{v}_4 = (-1,-1,1) /\sqrt{3}$, we define, for $s,t=1, \ldots, 4$:
\begin{equation}
\label{Eq:InputStates}
\ketbra{\varphi_s}{\varphi_s} = \frac{\mathbbm{1} + \vec{v}_s \cdot \vec{\sigma}}{2}, \quad \ketbra{\psi_t}{\psi_t} = \frac{\mathbbm{1} + \vec{v}_t \cdot \vec{\sigma}}{2}.
\end{equation}
The resulting correlations are then:
\begin{eqnarray}
P_{\ket{\Psi^-}}(a,b \, | \, \ket{\varphi_s},\ket{\psi_t}) = \left\{ \!
\begin{array}{cl}
\! \dfrac{2-(a+b)}{4} & \!\! \text{ if } s=t \\[0.3cm]
\! \dfrac{7-5a-5b+4ab}{12} & \!\! \text{ otherwise}.
\end{array}
\right.
\label{Eq:singlet_correlations}
\end{eqnarray}

This means that Alice and Bob must {\it  never} output $a=b=1$ when their input states are identical, but must otherwise produce this combination of outputs with some non-zero probability. Recall that Alice and Bob only have access to their respective quantum states $\ket{\varphi_s}$ and $\ket{\psi_t}$ and not their classical labels $s$ and $t$. Since these input states are not linearly independent, they cannot be unambiguously distinguished from one another~\cite{Chefles:1998hd}, i.e. Alice and Bob are bound to make mistakes if they try to guess the classical label $s$ and $t$ in some rounds of the simulation; thus the task of reproducing the correlations in Eq.~\eqref{Eq:singlet_correlations} cannot be achieved perfectly if Alice and Bob only make use of shared randomness and classical communication.

\subsection{Bell-like inequalities for any entangled state}
The no-go result given in the previous paragraph is in fact not specific to pure two-qubit maximally-entangled states, but rather is a general feature of all entangled states, as we shall demonstrate by constructing an explicit Bell-like inequality in the following Theorem and proof:

\begin{thm}
\label{thm1}

For any entangled quantum state $\rho$ of dimension $d_A \times d_B$, and sets of inputs states $\{ \ket{\varphi_s} \}_{s=1}^{d_A^2}$, $\{ \ket{\psi_t} \}_{t=1}^{d_B^2}$ whose corresponding density matrices are informationally complete (i.e. span the space of linear operators acting on $\mathbbm{C}^{d_A}$, $\mathbbm{C}^{d_B}$ respectively), there exist real coefficients $\beta^\rho_{st}$ defining a Bell-like inequality
\begin{equation}
\label{Eq:InequalityRho}
I^\rho(P) := \sum_{st} \ \beta^\rho_{st} \ P(1,1 \, | \, \ket{\varphi_s},\ket{\psi_t}) \ \ge \ 0,
\end{equation}
which is satisfied by all correlations obtainable from LOCC, but is violated by the quantum correlation $P_\rho(a,b \, | \, \ket{\varphi_s},\ket{\psi_t})$ obtained from the entangled state $\rho$, when Alice and Bob perform projections onto a maximally entangled state as in~\eqref{Eq:BellMeasurementAB}.

\end{thm}

\begin{cor}
Any entangled state can generate correlations that cannot be simulated by classical resources in a Quantum Input simulation task.
\end{cor}

\begin{proof}[Proof of Theorem~\ref{thm1}]
We use the fact that for any entangled state $\rho$, there exists a positive (but not completely positive) map $\Lambda$, such that $\one \otimes \Lambda (\rho)$ is not positive~\cite{Horodecki:1996separability}; it has at least one eigenstate $\ket{\xi}$ with negative eigenvalue $\lambda < 0$.
Let us denote by $\Lambda^{\!*}$ the (positive) dual map of $\Lambda$ (defined such that $\text{tr}[M \cdot \Lambda^{\!*}(N)] = \text{tr}[\Lambda(M) \cdot N]$ for all linear operators $M$, $N$), and let us decompose the Hermitian operator $\one \otimes \Lambda^{\!*} (\ketbra{\xi}{\xi})$ as follows:

\begin{equation}
\label{Eq:defbetas}
\one \otimes \Lambda^{\!*} (\ketbra{\xi}{\xi}) = \sum_{st} \ \beta^\rho_{st} \ \ketbra{\varphi_s}{\varphi_s}^{\!\top\!} \! \otimes \! \ketbra{\psi_t}{\psi_t}^{\!\top\!} ,
\end{equation}
where $\beta^\rho_{st}$ are real coefficients. Note that such a decomposition is always possible\footnote{Informationally complete sets of input states are not always required. What is required in our proof is that at least one of the eigenstates $\ket{\xi}$ of $\one \otimes \Lambda (\rho)$ with negative eigenvalue is such that $\one \otimes \Lambda^{\!*} (\ketbra{\xi}{\xi})$ can be decomposed as in~\eqref{Eq:defbetas}; the information-completeness assumption guarantees that this is indeed possible, but this assumption is not always necessary. For example in dimensions $d_A \neq d_B$, the state $\ket{\xi}$ has a Schmidt decomposition~\cite{Schmidt:1907gt,Ekert:1995bw} of the form $\ket{\xi} = \sum_{k=1}^n \alpha_k \ket{kk}$ with $n \le \min(d_A, d_B)$ elements, providing a decomposition for $\one \otimes \Lambda^{\!*} (\ketbra{\xi}{\xi})$ using $n^2$ input states for both Alice and Bob, by restricting these input states to the relevant subspace of dimension $n$.}
since, by assumption, the Hermitian operators $\left \{ \ketbra{\varphi_s}{\varphi_s}^{\!\top\!} \right \}_{s=1}^{d_A^2}$, $\left \{ \ketbra{\psi_t}{\psi_t}^{\!\top\!} \right \}_{t=1}^{d_B^2}$ used in Theorem~\ref{thm1} form a basis for Hermitian operators acting respectively on $\mathbbm{C}^{d_A},\mathbbm{C}^{d_B}$. The coefficients $\beta^\rho_{st}$ introduced in~\eqref{Eq:defbetas} are used to define the linear combination $I^{\rho}(P)$ in~\eqref{Eq:InequalityRho}.

\medskip

\paragraph*{(i)}To show that inequality~\eqref{Eq:InequalityRho} holds true for all correlations $P(a,b \, | \, \ket{\varphi_s},\ket{\psi_t})$ that can be obtained from LOCC operations, we shall now recall that all classes of LOCC operations, distinguished by the amount and/or number of rounds of communication~\cite{Chitambar:2012te}, are included in the set of so-called separable operations~\cite{Rains:1997sep,Vedral:1998qent,Bennett:1999qnwe,Chitambar:2012te}. As a result, the set of correlations obtainable by performing arbitrary LOCC operations on the input states $\ket{\varphi_s},\ket{\psi_t}$ ({\em LOCC correlations}, for short) is included in the set of correlations obtainable from separable measurements on $\ket{\varphi_s},\ket{\psi_t}$. Our strategy is to prove that inequality~\eqref{Eq:InequalityRho} holds for all separable measurements, which in turn, implies that it also holds for all LOCC correlations.

Let then $P_\textup{SEP}(a,b \, | \, \ket{\varphi_s},\ket{\psi_t})$ be the correlations obtained from some separable measurements on $\ket{\varphi_s},\ket{\psi_t}$. In our QI simulation task, such correlations take the form of:

\begin{equation}
P_\text{SEP}(a,b \, | \, \ket{\varphi_s},\ket{\psi_t}) = \text{tr}[ \ketbra{\varphi_s}{\varphi_s} \!\otimes\! \ketbra{\psi_t}{\psi_t} \cdot \Pi_{ab} ],
\end{equation}
where the $\Pi_{ab}$ are some POVM elements corresponding to outcomes $a$ and $b$. The separability constraint implies that for each $a$ and $b$, these can be decomposed as $\Pi_{ab} = \sum_k \Pi_{ab}^{\text{A},k} \otimes \Pi_{ab}^{\text{B},k}$, where the $\Pi_{ab}^{\text{A},k}$ and $\Pi_{ab}^{\text{B},k}$ are positive operators~\cite{Rains:1997sep,Vedral:1998qent,Bennett:1999qnwe,Chitambar:2012te}. Thus:
\begin{align}
I^\rho(P_\text{SEP}) & =  \sum_{st} \ \beta^\rho_{st} \ P_\text{SEP}(1,1 \, | \, \ket{\varphi_s},\ket{\psi_t})  = \sum_k \sum_{st} \ \beta^\rho_{st} \ \text{tr}[ \ketbra{\varphi_s}{\varphi_s} \!\otimes\! \ketbra{\psi_t}{\psi_t} \cdot \Pi_{11}^{\text{A},k} \otimes \Pi_{11}^{\text{B},k} ] \nonumber \\
 &= \sum_k \ \text{tr}[ \one \!\otimes\! \Lambda^{\!*} (\ketbra{\xi}{\xi}) \cdot (\Pi_{11}^{\text{A},k} \otimes \Pi_{11}^{\text{B},k})^{\top} ] = \sum_k \ \text{tr}[ \ketbra{\xi}{\xi} \cdot (\Pi_{11}^{\text{A},k})^{\!\top} \!\!\otimes\! \Lambda\big((\Pi_{11}^{\text{B},k})^{\!\top}\big) ] \ \ge \ 0,
\end{align}
where the last inequality is due to the fact that for each $k$, $(\Pi_{11}^{\text{A},k})^{\!\top} \!\!\otimes\! \Lambda\big((\Pi_{11}^{\text{B},k})^{\!\top}\big)$ is a positive operator. 
This proves that inequality~\eqref{Eq:InequalityRho} indeed holds for separable measurements on $\ket{\varphi_s},\ket{\psi_t}$. Since, as recalled above, correlations obtained using LOCC are a subset of those obtained from separable measurements~\cite{Rains:1997sep,Vedral:1998qent,Bennett:1999qnwe,Chitambar:2012te}, then $I^\rho(P_\textup{LOCC}) \ge 0$ also holds for all LOCC correlations $P_\textup{LOCC}(a,b \, | \, \ket{\varphi_s},\ket{\psi_t})$.

\medskip

\paragraph*{(ii)}On the other hand, the quantum correlations $P_\rho(a,b \, | \, \ket{\varphi_s},\ket{\psi_t})$ can be computed using \eqref{Eq:P_QI_effective} and the effective POVM elements of~(\ref{Eq:POVM_A_canonical}--\ref{Eq:POVM_B_canonical}). One then obtains:
\begin{align}
I^\rho(P_\rho) \! & =  \sum_{st} \ \beta^\rho_{st} \ P_\rho(1,1 \, | \, \ket{\varphi_s},\ket{\psi_t}) = \sum_{st} \ \beta^\rho_{st} \ \text{tr}[\mathcal{A}^\textup{A}_{1|\ket{\varphi_s}} \otimes \mathcal{B}^\textup{B}_{1|\ket{\psi_t}} \cdot \rho] \nonumber \\
 & = \sum_{st} \ \frac{\beta^\rho_{st}}{d_A d_B} \ \text{tr}\left[ \ketbra{\varphi_s}{\varphi_s}^{\!\top\!} \!\otimes\! \ketbra{\psi_t}{\psi_t}^{\!\top\!} \cdot \rho\right] =  \frac{\text{tr}[\one \!\otimes\! \Lambda^{\!*} (\ketbra{\xi}{\xi}) \!\cdot\! \rho]}{d_A d_B} =  \frac{\sandwich{\xi}{\one \!\otimes\! \Lambda (\rho)}{\xi}}{d_A d_B} = \frac{\lambda}{d_A d_B} < 0 ,
\end{align}
violating inequality~\eqref{Eq:InequalityRho}, and  thus concluding the proof.
\end{proof}

Before providing explicit examples of such Bell-like inequalities~\eqref{Eq:InequalityRho}, we remark that they have, in general, two interesting properties.

First, let us emphasize that these inequalities always use sets of input states that render unambiguous quantum state discrimination impossible. As seen in Section~\ref{Sec:SimulationTasks}, QI simulation tasks offer richer scenarios specifically because they allow sets of non-orthogonal quantum input states. As also noted before, non-orthogonality by itself does not rule out the possibility to perform unambiguous state discrimination (USD)~\cite{Chefles:1998hd}. Now, if USD was possible, Alice (Bob) could learn the classical label $s$ (resp. $t$) of her state $\ket{\varphi_s}$ ($\ket{\psi_t}$) with non-zero probability. Then, for all $s$ and $t$, both Alice and Bob would obtain conclusive results in some rounds of the simulation task. Alice and Bob could then coordinate to output $a=b=1$ only when $(s,t)$ is known and $\beta_{st}$ is negative, filtering out non-negative contributions to the inequality. However, this strategy cannot be used against inequalities constructed in Theorem~\ref{thm1}: informationally complete sets of $d_A^2, d_B^2$ rank-1 projectors in dimension $d_A$, $d_B$ have corresponding states $\ket{\varphi_s}$, $\ket{\psi_t}$ linearly dependent in $\mathbbm{C}^{d_A}$, $\mathbbm{C}^{d_B}$, rendering USD impossible~\cite{Chefles:1998hd}.

We also note that the Bell-like inequality~\eqref{Eq:InequalityRho} can be used to certify the entanglement of any state that violates it (and in particular the entangled state $\rho$ for which it is constructed). It corresponds indeed to a Measurement-Device-Independent Entanglement Witness (MDI-EW), as defined in Ref.~\cite{GAP:2012mdiews}, in a scenario with (trusted) quantum inputs. Our proof here made use of positive but not completely positive maps, but can equivalently be based on the existence of entanglement witnesses~\cite{Horodecki:1996separability}, as in the proof of~\cite{GAP:2012mdiews}. Our present Theorem~\ref{thm1} thus implies that the MDI-EWs constructed in~\cite{GAP:2012mdiews} can be used to distinguish entangled states not only from separable states, but also from arbitrary LOCC resources.

\section{Non-simulability of entangled Werner states}
\label{Sec:WernerStates}
To provide explicit examples of Bell-like inequalities, we turn to bipartite Werner states. In dimension $d^2=2^2$, they are defined as mixtures of the singlet state and the maximally mixed state:
\begin{equation}
  \label{Eq:WernerStates2}
  \rho_2 = v \ \ketbra{\Psi^-}{\Psi^-} + (1-v) \ \mathbbm{1}/4, \quad v \in [-\tfrac{1}{3},1].
\end{equation}
and in arbitrary dimension $d^2$ as~\cite{Werner:1989zz,Barrett:2002gu}:
\begin{equation}
  \label{Eq:WernerStates}
  \rho_d = v \frac{\mathbbm{1}-F}{d(d-1)} + (1-v)\frac{\mathbbm{1}}{d^2}, \quad v \in \left[-\frac{d-1}{d+1},1\right],
\end{equation}
where the flip operator $F$ is $F=\sum_{ij} \ketbra{ij}{ji}$.

Werner states are entangled if and only if~\cite{Werner:1989zz} $v > \frac{1}{d+1}$ (for qubits $v > \frac{1}{3}$). For some values of $v$, entangled Werner states are local, in the sense that their correlations in standard Bell scenarios can be simulated using solely shared randomness: this is the case for correlations obtained from projective measurements~\cite{Werner:1989zz} for $v \le 1-\tfrac{1}{d}$ (qubits: $v \le \frac{1}{2}$), and for correlations obtained from POVMs~\cite{Barrett:2002gu} for $v\le \tfrac{3d-1}{d^2-1}\left(1-\tfrac{1}{d}\right)^d$ (qubits: $v \le \frac{5}{12}$). 
Let us also recall that correlations from entangled Werner states can be simulated in standard Bell scenarios by using only finite communication on average~\cite{Massar:2001ug}; remarkably, for qubits, only a single bit is required in the worst case~\cite{Toner:2003communication}.

In QI scenarios, in contrast, all entangled Werner states exhibit correlations that cannot be simulated using LOCC, as we shall see explicitly below, first for qubits and then in higher dimensions.

\subsection{Two-qubit Werner states}

When Alice and Bob share the state $\rho_2$, and use the inputs given in~\eqref{Eq:InputStates} corresponding to the measurements specified in Eqs.~\eqref{Eq:POVM_A_canonical}-\eqref{Eq:POVM_B_canonical}, their correlations are a mixture
\begin{equation}
  \label{Eq:WernerCorrelation2}
P_{\rho_2}(a,b \, | \, \ket{\varphi_s},\ket{\psi_t}) =  \, v \ P_{\ket{\Psi^-}}(a,b \, | \, \ket{\varphi_s},\ket{\psi_t})  + \, (1\!-\!v) \  P_0(a,b \, | \, \ket{\varphi_s},\ket{\psi_t})
\end{equation}
of the singlet correlations given in~\eqref{Eq:singlet_correlations} and noise  $P_0(a,b \, | \, \ket{\varphi_s},\ket{\psi_t}) = (3-2a)(3-2b)/16$.

Entangled Werner states have non-positive partial transposes~\cite{Peres:1996separability}; hence, the map $\Lambda$ introduced in our construction above can simply be taken to be the transposition (which is self-dual).
The four eigenvalues of the partial transpose $\rho_2^{\top_{\!\!B}}$ of $\rho_2$ are $\{ \frac{1-3v}{4}, \frac{1+v}{4}, \frac{1+v}{4}, \frac{1+v}{4} \}$, with the eigenvalue $\lambda = \frac{1-3v}{4}$ corresponding to the eigenvector $\ket{\xi_2} = \ket{\Phi_2^+} $. $\rho_2^{\top_{\!\!B}}$ thus has a negative eigenvalue $\lambda < 0$ if and only if $v > \frac{1}{3}$, which indeed corresponds to the necessary and sufficient condition for $\rho_2$ to be entangled. We solve~\eqref{Eq:defbetas} to obtain the $\beta^{\rho_2}_{st}$, with which we compute the value of the inequality~\eqref{Eq:InequalityRho} for the correlations~\eqref{Eq:WernerCorrelation2}:
\begin{eqnarray}
\label{Eq:CoefficientsWerner2}
\beta^{\rho_2}_{st} = \frac{6 \delta_{st} - 1}{8}, \quad I^{\rho_2}(P_{\rho_2}) = \frac{1-3v}{16},
\end{eqnarray}
making use of the Kronecker delta $\delta_{st}$. One obtains a violation $I^{\rho_2}(P_{\rho_2}) < 0$ of the Bell-like inequality~\eqref{Eq:InequalityRho} whenever $v > \frac{1}{3}$, i.e. for all entangled two-qubit Werner states, while all LOCC correlations satisfy $I^{\rho_2}(P_\text{LOCC}) \ge 0$: entangled two-qubit Werner states correlations~\eqref{Eq:WernerCorrelation2} cannot be simulated by classical resources in QI scenarios.

Note that the Bell-like inequality thus obtained is exactly the same as the first MDI-EW derived in Ref.~\cite{GAP:2012mdiews} to certify the entanglement of 2-qubit Werner states.

\subsection{Higher-dimensional Werner states}

The two-qubit example above generalizes readily to the $d \times d$ dimensional case. To simplify our computations, we consider input states $\left\{ \ket{\varphi_s} \right\}_{s=1}^{d^2}$ satisfying the requirement that 
\begin{equation}
\label{Eq:SICPOVM}
	\left| \braket{\varphi_{s'}}{\varphi_s}\right|^2 = \frac{\delta_{s'\!s} \, d + 1}{d+1},
\end{equation} 
and the same inputs for Bob, i.e., $\ket{\psi_t} = \ket{\varphi_t}$.
Condition~\eqref{Eq:SICPOVM} is satisfied by the quantum input states of Eq.~\eqref{Eq:InputStates} for $d=2$; sets of such states are given, mostly numerically, in Ref.~\cite{Scott:2010ew} for $d \le 67$, and are conjectured to exist for all dimensions~\cite{Scott:2010ew}. Using the identity
\begin{equation}
\label{Eq:SwapTraceIdentity}
\tr(F\cdot A\otimes B)=\tr(A\cdot B)
\end{equation}
for any $d\times d$ matrices $A$ and $B$, one can show that for $\rho_d$ and the measurements specified in Eqs.~\eqref{Eq:POVM_A_canonical}-\eqref{Eq:POVM_B_canonical}, one gets, from Eq.~\eqref{Eq:P_QI_effective}:
\begin{equation}
  \label{Eq:WernerCorrelation}
P_{\rho_d}(a,b \, | \, \ket{\varphi_s},\ket{\psi_t})\!=\!\frac{(d^2\!\!-\!\!1)^{3-\!a\!-\!b}+(\!-1)^{a\!+\!b} (1\!-\!d^2 \delta_{st})v}{d^4(d^2\!\!-\!\!1)}.
\end{equation}

The positive map $\Lambda$ can again be taken to be the transposition.
To obtain $\beta^{\rho_d}_{st}$ from~\eqref{Eq:defbetas}, we note that for every entangled $\rho_d$, i.e., for $v > \frac{1}{d+1}$, its partial transpose
\begin{equation}
\rho_d^{\top_{\!\!\text{B}}} = \left[ 1 + \frac{v}{d-1} \right] \frac{\mathbbm{1}}{d^2} - \frac{v}{d-1} \ketbra{\Phi_d^+}{\Phi_d^+}
\end{equation}
has a negative eigenvalue  $\lambda=\frac{1-v(d+1)}{d^2}$, corresponding to the  eigenstate $\ket{\xi_d} = \ket{\Phi_d^+}$.
Solving Eq.~\eqref{Eq:defbetas} with $\ketbra{\xi_d}{\xi_d}^{\!\top_{\!\!\text{B}}} = \frac{F}{d}$, we get 
\begin{eqnarray}
\label{Eq:CoefficientsWerner}
\beta^{\rho_d}_{st} = \frac{d(d+1)\delta_{st} - 1}{d^3},
\end{eqnarray}
which gives $I^{\rho_d}(P_{\rho_d}) = \frac{\lambda}{d^2}=\frac{1-v(d+1)}{d^4}<0$. Hence, one obtains a violation $I^{\rho_d}(P_{\rho_d})  < 0$ of the Bell-like inequality~\eqref{Eq:InequalityRho} whenever $v > \frac{1}{d+1}$, i.e., for all entangled Werner states.

\section{Conclusion}
Inspired by Buscemi's generalization of standard Bell scenarios, where the parties are provided external quantum input states, we have introduced the notion of quantum input simulation tasks. Within the framework of such tasks, we showed that if separated parties only have access to shared randomness and classical communications, the set of correlations that they can produce must satisfy some Bell-like inequalities that can however be violated by entangled quantum states, thus establishing a new feature of quantum entanglement. Our no-simulation result thus improves over those derived from standard Bell scenarios in two aspects: firstly, our inequalities can be violated by entangled states which do not violate any Bell inequality at the single-copy level. Secondly, Bell inequalities can always be violated using classical communication~\cite{Toner:2003communication,Degorre:2005simulating,Massar:2001ug}, whereas such classical resources never violate our inequalities.

In experimental setups, such a construction can be used to check that Alice and Bob share quantum resources (be it an entangled state or a quantum channel) when their measurements are not performed in a space-like separated way. The robustness of our scheme against imperfections in the preparation of the quantum inputs is left for future work; a brief discussion of the effect of losses can be found in~\cite{GAP:2012mdiews}.

From an information-theoretic perspective, our results reinforce the point that {\it all} entangled states have an edge over classical resources in their information processing capacities, in particular in the context of certain bipartite simulation tasks as considered here. Given the rich structure of multipartite entanglement~\cite{Horodecki:2009gb}, it would also certainly be desirable to understand how our results generalize to the multipartite scenario.

\subsection*{Acknowledgments}
We acknowledge useful discussions with Francesco Buscemi, Andreas Winter and Samuel Portmann, and are grateful to Antonio Ac\'in for suggesting the extension of our original result to all entangled states by using positive, but not completely positive maps. This work is supported by the Swiss NCCR ``Quantum Science and Technology",  the CHIST-ERA DIQIP, the European ERC-AG QORE and a UQ Postdoctoral Research Fellowship.

\end{document}